\documentclass[conference,10pt]{IEEEtran}
\IEEEoverridecommandlockouts
\usepackage{cite}
\usepackage{amsmath,amssymb,amsfonts}
\usepackage{algorithmic}
\usepackage{graphicx}
\usepackage{textcomp}
\usepackage{xcolor}
\usepackage{amsthm}

\newtheorem{prop}{Proposition}
\newtheorem{assumption}{Assumption}
\newtheorem{coro}{Corollary}

\def\BibTeX{{\rm B\kern-.05em{\sc i\kern-.025em b}\kern-.08em
    T\kern-.1667em\lower.7ex\hbox{E}\kern-.125emX}}
\begin{document}

\def\Ab{\mathbf{A}}
\def\Bb{\mathbf{B}}
\def\Qb{\mathbf{Q}}
\def\Db{\mathbf{D}}
\def\Fb{\mathbf{F}}
\def\xb{\mathbf{x}}
\def\ub{\mathbf{u}}
\def\zb{\mathbf{z}}
\def\nb{\mathbf{n}}

\def\xbhat{\hat{\xb}}

\title{Goal-Oriented State Information Compression for Linear Dynamical System Control\\
}


\author{\IEEEauthorblockN{Li Wang\IEEEauthorrefmark{1}\IEEEauthorrefmark{2},
Chao Zhang\IEEEauthorrefmark{1}\IEEEauthorrefmark{2},
Samson Lasaulce\IEEEauthorrefmark{1}\IEEEauthorrefmark{3}, 
Lina Bariah\IEEEauthorrefmark{1}\IEEEauthorrefmark{4}, and
Merouane Debbah\IEEEauthorrefmark{1}}
\IEEEauthorblockA{\IEEEauthorrefmark{1}6G Research Center, Khalifa University, Abu Dhabi, UAE}
\IEEEauthorblockA{\IEEEauthorrefmark{2}Central South University, Changsha, China}
\IEEEauthorblockA{\IEEEauthorrefmark{3}University of Lorraine, Nancy, France}
\IEEEauthorblockA{\IEEEauthorrefmark{4}Open Innovation AI, Abu Dhabi, UAE}
\vspace{-12pt}}
\maketitle

\begin{abstract}
In this paper, we consider controlled linear dynamical systems in which the controller has only access to a compressed version of the system state. The technical problem we investigate is that of allocating compression resources over time such that the control performance degradation induced by compression is minimized. This can be formulated as an optimization problem to find the optimal resource allocation policy. Under mild assumptions, this optimization problem can be proved to have the same well-known structure as in \cite{Vaidyanathan-book-2006}, allowing the optimal resource allocation policy to be determined in closed-form. The obtained insights behind the optimal policy provide clear guidelines on the issue of "when to communicate" and "how to communicate" in dynamical systems with restricted communication resources. The obtained simulation results confirm the efficiency of the proposed allocation policy and illustrate the gain over the widely used uniform rate allocation policy.

\end{abstract}


\section{Introduction}

The rapid growth of data-driven technologies and widespread use of smart IoT devices has led to a massive increase of data volumes, challenging future wireless networks to enhance transmission rates, energy efficiency, latency, and reliability. As network capacity expansion lags behind data growth, revisiting data compression techniques becomes crucial for efficient data transmission and resource management. 

Recent developments in semantic and goal-oriented compression methods have become key to addressing increasing data volumes \cite{Zhang-IoTM-2022}\cite{Deniz-JSAC-2023}. Unlike traditional compression, which focuses on minimal data loss, these methods align compression with specific tasks, maintaining quality and essential features for objectives such as real-time decision-making \cite{Zhang-AE-2021}, efficient resource allocation \cite{Zhang-Tcom-2023}, and inference \cite{Chen-Tc-2023}, among others. For instance, \cite{Hang-JSAC-2023} considers the final use of the data to the system task modeled by an optimization problem and utilizes high-resolution quantization theory to establish the effectiveness of a goal-oriented quantization scheme. The authors of \cite{Niu-arxiv-2024} proposed a synonymous mapping to describe the correlation between the semantic information and the original source data, reducing coding length in both lossy and lossless scenarios. Given the abstract nature of semantic information, deep learning techniques are extensively employed to extract relevant semantic content and key task features, employing methods such as transformer-based joint-source encoder-decoders \cite{Qin-TSP-2021} and generative models \cite{Barbarossa}. Most existing studies focus on static systems with independent and identically distributed (i.i.d.) observation processes, emphasizing data meaning preservation over managing dynamic tasks, rather than handling tasks that involve dynamic or time-varying random processes. As the next-generation communication systems are expected to remotely control a wide range of dynamic systems, including cars, robots, and drones, optimizing the use of communication resources is essential to ensure reliable and efficient wireless connectivity for these advanced applications. 

In dynamical systems, actions have long-lasting effects on future states, often sacrificing immediate rewards for long-term benefits. The information to be compressed, including system states and observations, typically exhibits memory and can be modeled with linear recursive equations, Markov decision processes (MDP), or more complex structures. This feature highlights the importance of temporal correlation in the development of control policies. Various studies have addressed the dynamic aspects of systems; for instance, \cite{Kui-wiopt-2023} proposed adaptive power control policies in linear systems for vehicle platooning, while \cite{Vineeth} introduces a threshold-based communication scheme that leverages a waiting time policy for energy-efficient wireless control of nonlinear systems. In MDP settings, \cite{Pappas} develops a goal-oriented sampling policy to minimize actuation errors in real-time tracking under resource limitations, and \cite{Mostaani} explores the integration of quantization and control in bit-budgeted communications. Despite these advancements, the impact of compression noise on dynamical systems remains under-explored, and addressing this gap is a primary focus of our current research.

In this paper, we propose a goal-oriented compression framework for optimal resource allocation, focusing on the effects of compression noise. We analyze the noise impact using the linear quadratic regulator (LQR) problem, employing a non-i.i.d. system state modeled by a recursive equation. Utilizing Shannon's rate-distortion theory, we establish a relationship between noise variance and transmission rate, deriving a closed-form expression for the optimal rate that decreases over time, indicating a higher initial resource allocation in an invariant system. The presented simulation results further validate the efficiency of our proposed method.

\section{Problem formulation}
\label{sec:problem formulation}

The rate-limited control setting discussed in this article features an LQR system with a state vector $\xb_t \in \mathbb{R}^n$ and a communication link between an observer and the controller, as depicted in Fig. \ref{fig:system}.
\vspace{-0.3cm}
\begin{figure}[htbp]
\centering
\includegraphics[width=0.8\linewidth]{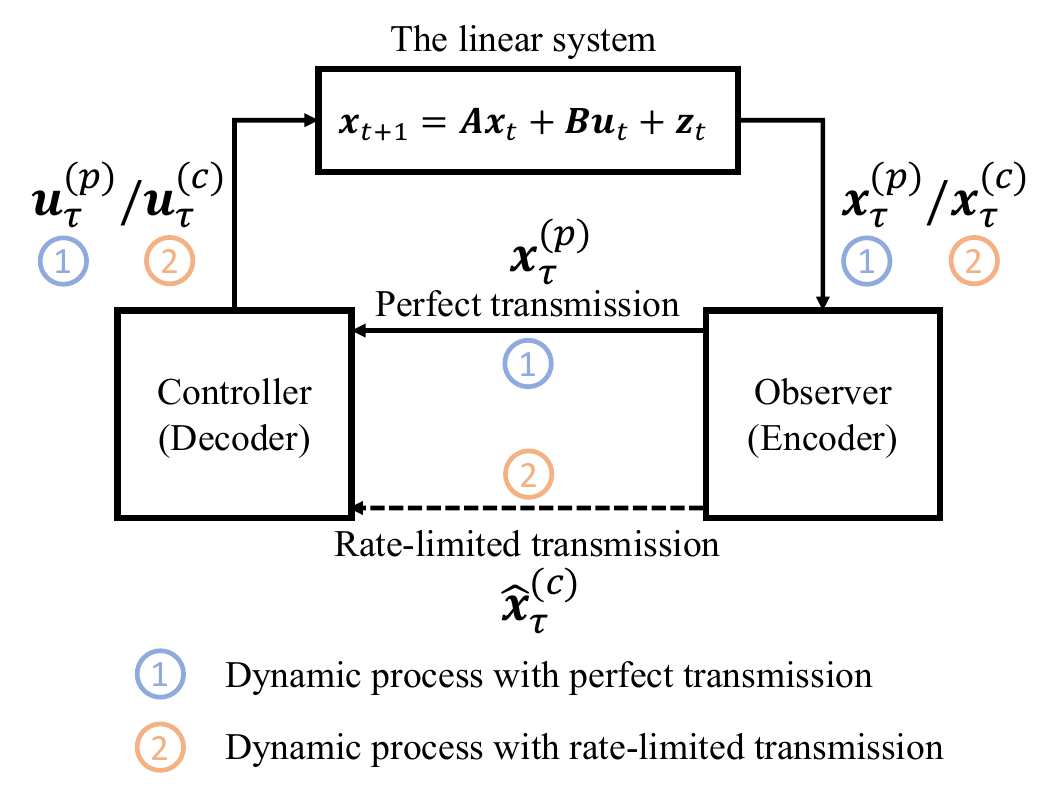}
\caption{LQR setting with a communication link (the dashed line) from the observer to the controller.}
\label{fig:system}
\end{figure}

\vspace{-0.1cm}
To evaluate performance degradation due to limited transmission rates, it's crucial to analyze state deviations from imperfect communication. We compare a control system with perfect communication to one with rate-limited communication, denoting their state-control pairs as $(\xb^{(p)}_{\tau},\ub^{(p)}_{\tau})$ and $(\xb^{(c)}_{\tau},\ub^{(c)}_{\tau})$, respectively.
\subsection{Perfect transmission}
The linear dynamical system with perfect communication link is described by:
\begin{equation}
\mathbf{x}^{(p)}_{t+1}=\Ab\xb^{(p)}_t+\Bb\ub^{(p)}_t+\zb_t,
\end{equation}
where $\Ab\in \mathbb{R}^{n\times n}$ and $\Bb \in \mathbb{R}^{n\times m}$ are system matrices, $\xb^{(p)}_t \in \mathbb{R}^n$ is the state, $\ub^{(p)}_t \in \mathbb{R}^m$ is the control input and $\zb_0,\zb_1,\dots$ are i.i.d. disturbances with zero mean and covariance $\Sigma_z \in \mathbb{S}^n_{++}$, where $\mathbb{S}^n_{++}$ denotes the set of all positive definite $n \times n$ matrices.

For a finite time system, the design of the control policy in a typical control system with perfect transmission aims to minimize the average-stage LQR cost upon reaching the horizon $T$:
\begin{equation}
\min_{\ub^{(p)}_0,\ub^{(p)}_1,\dots} J^{(p)} \triangleq \frac{1}{T}\mathbb{E}\left[ \sum_{t=0}^{T-1} \xb^{(p)'}_t\Qb\xb^{(p)}_t+\ub^{(p)'}_t\Db\ub^{(p)}_t \right],
\label{eq:system policy criterion}
\end{equation}
where $x'$ represents the transpose, and $\Qb\in \mathbb{S}^n_{++}$ and $\Db\in \mathbb{S}^n_{++}$ are cost matrices. The expectation is calculated with respect to the disturbance process $\zb_t$.
It has been shown that to solve (\ref{eq:system policy criterion}), the optimal strategy is to use control actions which are proportional to the current system state as follows
\begin{equation}
    \ub^{(p)}_t=\Fb_t\xb^{(p)}_t,
\end{equation}
where $\Fb_t\in\mathbb{R}^{m\times n}$ is the control matrix and can be computed as in \cite{shaiju2008formulas}. By inserting $\ub^{(p)}_t$ into (\ref{eq:system policy criterion}), $\Fb_t,~0\leq t\leq T$ are estimated by:
\begin{equation}
    \min_{\Fb_0,\Fb_1,\dots}J^{(p)} \triangleq \frac{1}{T}\mathbb{E}\left[ \sum_{t=0}^{T-1} \xb^{(p)'}_t( \Qb+\Fb_t'\Db\Fb_t )\xb^{(p)}_t \right],
\end{equation}
where $T$ denotes the transpose of matrix.

\subsection{Rate-limited transmission}
In a rate-limited system, the observed system state is transmitted imperfectly. We denote $\nb_t\in\mathbb{R}^n$ as the deviation between the real system state and the compressed system state as follows:
\begin{equation}
    \hat{\xb}^{(c)}_t = \xb^{(c)}_t + \nb_t^{(c)},~~\forall~0\leq t\leq T,
    \label{eq:compression of state}
\end{equation}
where 
\begin{equation}
    \nb_t^{(c)} \sim p(R_t),~~\forall~0\leq t\leq T.
\end{equation}
Here $p(R_t)$ represents a probability density function (p.d.f.) that is a function of the transmission rate $R_t$ (rate allocated at time $t$). A higher transmission rate $R_t$ often yields a lower value of $\nb_t$, leading to a smaller variance of the distribution $p(R_t)$. With the compression noise, biased control inputs $\ub^{(c)}_t$ and biased respective states $\xb^{(c)}_t$ are derived, according to the dynamic process, as follows
\begin{align}
    \ub^{(c)}_t &= \Fb_t\hat{\xb}^{(c)}_t, \notag \\
    \xb^{(c)}_{t+1} &= \Ab\xb^{(c)}_t + \Bb\ub^{(c)}_t + \zb_t,
\end{align}
where $\ub^{(c)}_t = \Fb_t\hat{\xb}^{(c)}_t$, and $\hat{\xb}^{(c)}_t$ is the compressed observation of the state at time $t$, as defined in (\ref{eq:compression of state}).
The average-stage LQR cost then become
\begin{equation}
    J^{(c)}(R_0,R_1,\dots,R_T) \triangleq \frac{1}{T}\mathbb{E}\left[ \sum_{t=0}^{T-1} \xb^{(c)'}_t\Qb\xb^{(c)}_t+\ub^{(c)'}_t\Db\ub^{(c)}_t \right]
    \label{eq:J_c}
\end{equation}

It can be noticed from \ref{eq:J_c} that the compression noise added at time $t$ will have influence on the LQG cost. It is also worthy to note that the noise intensity at time $t$ is influenced by the allocated rate at that time. When the allocated rate is unlimited, the compression noise is negligible.
However, in a linear dynamical system with rate-limited communication, the rate allocated at each time is limited by the available resources. In this work, we consider the scenario where the total rate budget is limited, i.e., the sum of rate should be smaller than a threshold, namely $\sum_{t=0}^T R_t < R_{\mathrm{sum}}$, where $R_t,~0\leq t\leq T$ is the allocated rate at time $t$, and $R_{\mathrm{sum}}$ represent the upper bound of the overall rate.

The goal of the resource allocation scheme is to minimize the LQR costs caused by the compression process, while ensuring that the transmission rate is within the available rate budget. More precisely, the rate allocation scheme can be formulated as the following constrained optimization problem
\begin{gather}
     \min_{(R_0,R_1,\dots,R_T)} J(R_0,R_1,\dots,R_T) \triangleq J^{(c)} - J^{(p)}, \notag \\
     \mathrm{s.t.}~~\sum_{t=0}^T R_t \leq R_{\mathrm{sum} }   
    \label{eq:diff_J}
\end{gather}

\section{Analysis of the rate allocation problem}
For notational convenience and calculation simplification, we study the standard setup of discrete LQR problems as in Assumption 1. However, we would like to claim that the way to derive the rate allocation policy for vector state systems is quite similar to scalar state systems.

\begin{assumption}
The system state and control action are both scalars.
\end{assumption}

According to this assumption, we rewrite the perfect communication system as follows \vspace{-3pt}
\begin{align}
x_{t+1}^{(p)} &= Ax_t^{(p)} + Bu_t^{(p)} + z_t, \notag \\
u_t^{(p)} &= F_t \cdot x_t^{(p)},
\label{eq:LTI sys}
\end{align}

Assuming that the controller coefficients are optimized and fixed, then a system with a  rate-limited communication can be given by \vspace{-3pt}
\begin{align}
x_{t+1}^{(c)} &= Ax_t^{(c)} + Bu_t^{(c)} + z_t, \notag \\
u_t^{(c)} &= F_t \cdot \hat{x}_t^{(c)}, \notag \\
\hat{x}_t^{(c)} &= x_t^{(c)} + n_t^{(c)},
\label{eq:LTI sys with compression}
\end{align}

 In order to quantify the performance degradation caused by rate-limited communications in LQR problems, it is essential to evaluate the biases between system states with and without noise, and this is the motivation behind the following proposition.
\begin{prop}
    For a rate-limited invariant dynamical system, the state at time $t+1$ can be derived from the state in a perfect-transmission system as
    \begin{equation}
        x_{t+1}^{(c)} = x_{t+1}^{(p)} + \sum_{m=0}^{t} \left(\prod_{i=0}^{t-m}(A*\mathrm{sgn}(i)+BF_{m+i}) \right) n_{m}^{(c)},
    \end{equation}
    for $0\leq t\leq T$, where $x_{t+1}^{(p)}$ and $x_{t+1}^{(c)}$ represent respectively the state in a perfect-transmitted system, and the state with compressed transmission, at time $t+1$. $\mathrm{sgn(\cdot)}$ is the sign function.
\end{prop}

\begin{proof}
From \eqref{eq:LTI sys} and \eqref{eq:LTI sys with compression}, we can derive
\begin{align}
    x_{t+1}^{(p)} &= (A+BF_t)x_t^{(p)} + z_t \label{eq:AR of perfect system}\\
    x_{t+1}^{(c)} &= (A+BF_t)x_t^{(c)} + z_t + (BF_t) n_t^{(c)},\label{eq:AR of compressed system}
\end{align}
for $0\leq t\leq T$. By subtracting \eqref{eq:AR of perfect system} from \eqref{eq:AR of compressed system}, the auto-regressive function can be derived as
\begin{equation}
    x_{t+1}^{(c)} - x_{t+1}^{(p)} = (A+BF_t)(x_t^{(c)}-x_t^{(p)}) + (BF_t) n_t^{(c)}.
\end{equation}

Note that the first state is not influenced by a compressed control policy, i.e., $x_0^{(c)} = x_0^{(p)}$, thus we can obtain
\begin{align}
x_{t+1}^{(c)} - x_{t+1}^{(p)} = \sum_{m=0}^{t} \left(\prod_{i=0}^{t-m}(A*\mathrm{sgn}(i)+BF_{m+i}) \right) n_{m}^{(c)}. \notag
\end{align}
By applying transposition of terms, we obtain
\begin{equation}
    x_{t+1}^{(c)}  = x_{t+1}^{(p)} + \sum_{m=0}^{t} \left(\prod_{i=0}^{t-m}(A*\mathrm{sgn}(i)+BF_{m+i}) \right) n_{m}^{(c)}.
    \label{eq:relation of xc and xp}
\end{equation}
\end{proof}

While the differences in system states can be explicitly described, assessing the impact of compression noise on system costs is challenging due to the complex nature of noise. To address this, we make specific assumptions to simplify the theoretical analysis. 
\begin{assumption}
The compression noise is assumed to be independent of the signals and follows the Gaussian distribution, namely,
\begin{equation}
n_t^{(c)} \sim \mathcal{N}(0, c^2 2^{-2R_t}),
\label{assumption: compression noise Gaussian distribution}
\end{equation}
where $R_t$ is the rate allocated to transmit the state at time slot $t$, and $c$ is a compression-related constant.
\end{assumption}

Based on this assumption, the system state bias could be further derived in the following proposition.
\begin{prop}
When the compression noise satisfies Assumption 2, the system state can be expressed as
\begin{equation}
    x_{t+1}^{(c)} =  x_{t+1}^{(p)} + \bar{n}_{t+1}^{(c)},
\end{equation}
where
\begin{equation}
    \bar{n}_{t+1}^{(c)} \sim \mathcal{N}(0, \sigma_{t+1}^2)
\end{equation}
and
\begin{equation}
    \sigma_{t+1}^2 =  c^2\sum_{m=0}^{t}  \left(\prod_{i=0}^{t-m}(A*\mathrm{sgn}(i)+BF_{m+i})\right)^2 2^{-2R_m}
\end{equation}
\label{proposition:ralation between xc and xp}
\end{prop}
\begin{proof}
Based on Eq.\ref{eq:relation of xc and xp} we have
\[
x_{t+1}^{(c)} =  x_{t+1}^{(p)} + \bar{n}_{t+1}^{(c)}
\]
where
\[
\bar{n}_{t+1}^{(c)} = \sum_{m=0}^{t} \left(\prod_{i=0}^{t-m}(A*\mathrm{sgn}(i)+BF_{m+i}) \right) n_{m}^{(c)}.
\]
Knowing that the sum of independent Gaussian variables still follows Gaussian distribution with the expectation and variance equal to the sum of expectation and variance of each variable, and $n_{m}^{(c)}\sim \mathcal{N}(0, c^2 2^{-2R_m})$ according to Assumption \ref{assumption: compression noise Gaussian distribution}, straightforward proving Proposition \ref{proposition:ralation between xc and xp}.


\end{proof}

As system state bias induced by compression noise can be approximated by additive Gaussian noise, the rate allocation problem defined by \eqref{eq:diff_J} can be proved equivalent to a convex optimization problem, shown in Proposition \ref{proposition:constrained optimization problem.}.
\begin{prop}
    Based on Assumption 1 and Assumption 2,  the original rate allocation problem defined by (\ref{eq:diff_J}) can be equivalently solved by the following constrained optimization problem: 
    \begin{gather}
        \min_{(R_0,R_1,\dots,R_T)} \sum_{t=0}^T a_t 2^{-2R_t}, \notag\\
        \mathrm{s.t.}~~\sum_{t=0}^{T} R_t \leq R_{\text{sum}}, \notag
    \end{gather}
    where
    \small
    \begin{equation}
        a_t = c^2  \sum_{m=t+1}^{T} (Q+F_m^2D) \left(\prod_{i=0}^{m-1-t}(A*\mathrm{sgn}(i)+BF_{t+i})^2\right)   + F_t^2D c^2.
    \end{equation}
\label{proposition:constrained optimization problem.}
\end{prop}
\begin{proof}
The optimality loss at time slot \( t \) can be expressed as
\small
\begin{align*}
&\quad\,\,\mathbb{E}[f_t(x_t^{(c)}, u_t^{(c)}) - f_t(x_t^{(p)}, u_t^{(p)})] \\
&= \mathbb{E}[Q(2x_t^{(p)} + \bar{n}_t^{(c)})\bar{n}_t^{(c)} + F_t^2D(\bar{n}_t^{(c)}+ n_t^{(c)})(2x_t^{(p)} +\bar{n}_t^{(c)}+ n_t^{(c)})] \\
&= Q \sigma^2_{t} + F_t^2D (\sigma_{t}^2 + c^2 2^{-2R_t}) \\
&= (Q+F_t^2D) c^2\sum_{m=0}^{t-1}  \left(\prod_{i=0}^{t-1-m}(A*\mathrm{sgn}(i)+BF_{m+i})^2\right) 2^{-2R_m}\\
&\quad\quad+ F_t^2D c^2 2^{-2R_t}
\end{align*}

For the bit allocation problem, which can be formulated as
\[
\min \sum_{t=0}^{T} \mathbb{E}[f_t(x_t^{(c)}, u_t^{(c)})-f_t(x_t^{(p)}, u_t^{(p)})]
\]
subject to: \vspace{-5pt}
\[
\sum_{t=0}^{T} R_t \leq R_{\text{sum}}
\]
the objective function can be rewritten as
\[
 a_02^{-2R_0}+ a_1 2^{-2R_1} + a_2 2^{-2R_2} + \ldots + a_T 2^{-2R_T}
\]
where
\small
\[
a_t = c^2  \sum_{m=t+1}^{T} (Q+F_m^2D) \left(\prod_{i=0}^{m-1-t}(A*\mathrm{sgn}(i)+BF_{t+i})^2\right)   + F_t^2D c^2 
\]
for $1<t<T$, and $a_T=0$.
\end{proof}

The solution of the optimization problem can be obtained easily by using Lagrangian multiplier approach, and can be written as
\[
R_k^* = \frac{1}{2} \log \frac{a_k}{\sqrt[T]{\prod_{i=0}^{T-1} a_i}} + \frac{R_{\text{sum}}}{T}
\]
for $0\leq k<T$, and $R_T^*=0$.

\begin{coro}
In a system with an invariant controller $F_0=F_1=F_2=\cdots=F$, the optimal rate with a given overall budget \( R_t^* \) is a decreasing function of \( t \), namely
\begin{equation}
    R_0^*>R_1^* > R_{2}^* > \ldots > R_T^*,
\end{equation}
indicating that more transmission resource should be allocated to early stages in this case.
\end{coro}
This can be explained by the fact that the compression noise at early stages will be accumulated in subsequent stages, thus having smaller error in early stages yields a better performance.

\section{Numerical analysis}

In this section, we present simulation results to show the performance of the proposed rate allocation scheme in dynamic wireless systems.

\subsection{Evaluation metric}
A relative LQR cost, represented by $J_{RCost}$, is used to evaluate the performance of the proposed rate-allocation scheme, where
\begin{equation}
    J_{RCost} = (J^{(c)}-J^{(p)})/J^{(p)},
    \label{eq:relative cost}
\end{equation}
given that $J^{(p)}$ and $J^{(c)}$ are respectively the LQR cost of perfect-transmitted system and networked control system, defined in Eq. (\ref{eq:system policy criterion}) and Eq. (\ref{eq:J_c}). The goal of the rate-allocation scheme is to mitigate the additional cost caused by the transmission rate constraint, or equivalently, to minimize the relative LQR cost.

\subsection{Time-invariant systems}
The considered system parameters are defined as follows: $B = 1$, $Q = 2$, and $D = 5$, with $A$ varying between $[0,3]$. For each case, the optimal controller $F_t, 0\leq t\leq T$ is obtained according to the standard formula \cite{shaiju2008formulas}. This study models a 11-stage finite horizon system starting from $x_0=100$. Gaussian noise is added to each state $x_t$, with variance determined by (\ref{assumption: compression noise Gaussian distribution}).
The relative costs are shown in Figs.\ref{fig:cost analysis A}, while Fig. \ref{fig:rate analysis A} illustrates the optimized rate allocation scheme under various settings of $A$. We observe that with a lower $A$ value, the performance of the optimized rate-allocation scheme approaches the constant rate scheme, resulting in a smaller difference in relative LQR costs between the proposed and constant-cost schemes. Conversely, with a higher $A$ value, the superiority of the optimal rate allocation scheme becomes more pronounced, leading to significantly lower LQR costs.

This finding aligns with the following: when $A$ is small, previous states have less influence on the current state, thereby increasing the impact of the current controller. As a result, the noise from the earlier state transmissions has less effect on the overall cost, making the transmission noise for each state equally critical. Conversely, when previous states significantly influence the current state (larger $A$), more resources are dedicated to transmitting earlier states to ensure accuracy and minimize cost impacts.

\begin{figure}[!htbp]
    \centering
    \includegraphics[width=0.75\linewidth]{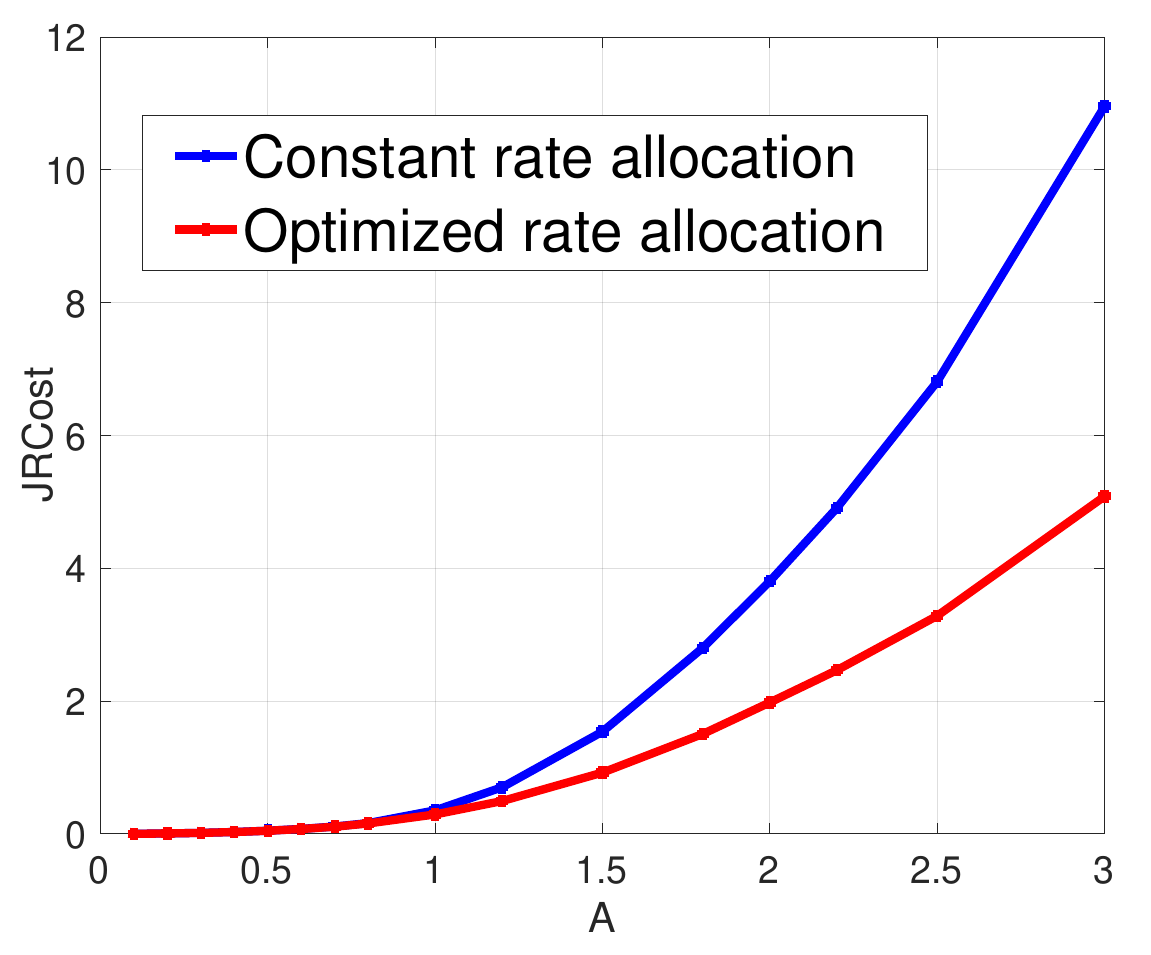}
    \caption{Relative LQR costs (\ref{eq:relative cost}), 
    The effectiveness of the proposed optimal rate allocation scheme is significantly enhanced when the dynamic process is more influenced by state transitions than by controller inputs (i.e., when the value of $A$ is higher), outperforming the constant-rate scheme more noticeably.}
    \label{fig:cost analysis A}
\end{figure}

\begin{figure}[!htbp]
    \centering
    \includegraphics[width=0.75\linewidth]{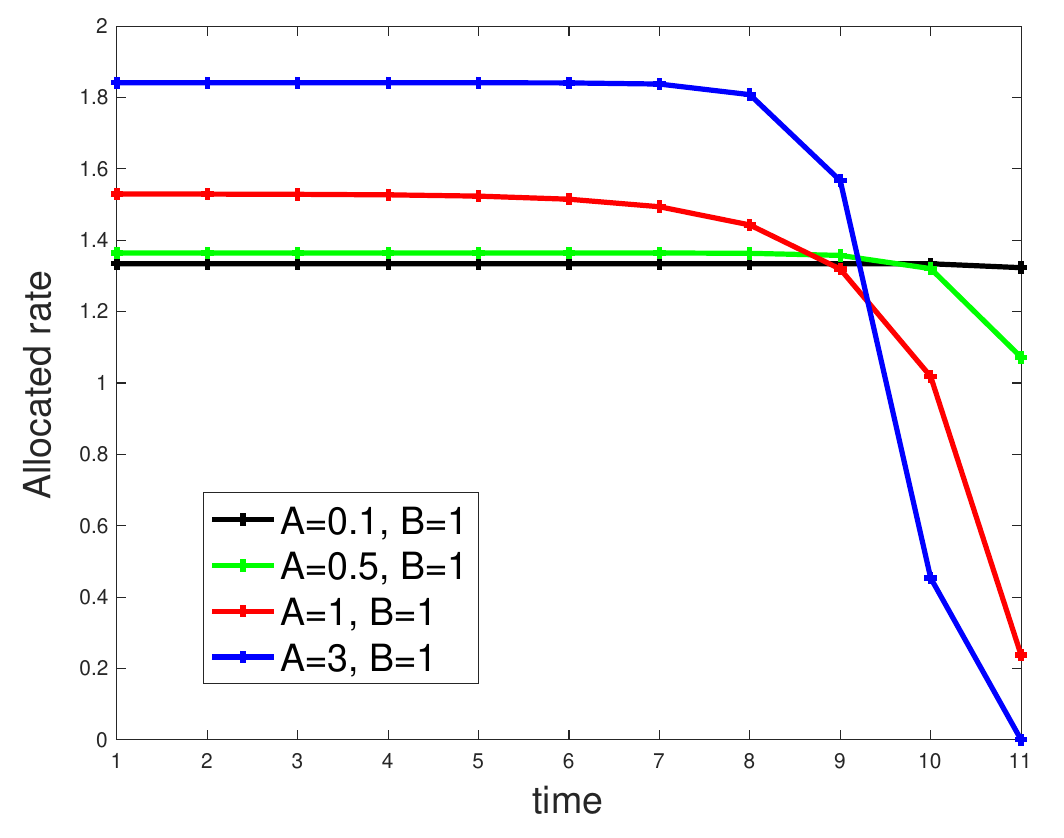}
    \caption{Optimal rate allocation for invariant linear dynamic systems with 11 stages, defined with different $(A,B)$ pairs. 
    It demonstrates that when the dynamical process is more influenced by state transitions than by controller inputs (i.e., when A is larger), more transmission resources should be allocated to the earlier states.}
    \label{fig:rate analysis A}
\end{figure}


\subsection{Time-variant dynamical systems}

Dynamical systems often experience changes due to external disturbances. For example, a control system might alter when faced with environmental perturbations. This paper shows that an optimized rate allocation scheme can also adapt to changes in a dynamical system. Here we show the optimized allocated rate for a variant dynamical system with a change characterized by a sudden shift in the system coefficient $A$.
Specifically, due to the computational complexity involved, we consider a system with 4 stages, undergoing a coefficient change at time $t=3$, switching from $A_1$ to $A_2$. 
Through exhaustive search, the optimal rate allocation for this variable system is determined, revealing that rate allocation should increase when there is a jump in 
$A$. More over, the allocated rate increases in the time before the jump appears.
This adaptation is illustrated in Fig. \ref{fig:rate analysis A variant system}.

\begin{figure}[!htbp]
    \centering
    \includegraphics[width=0.75\linewidth]{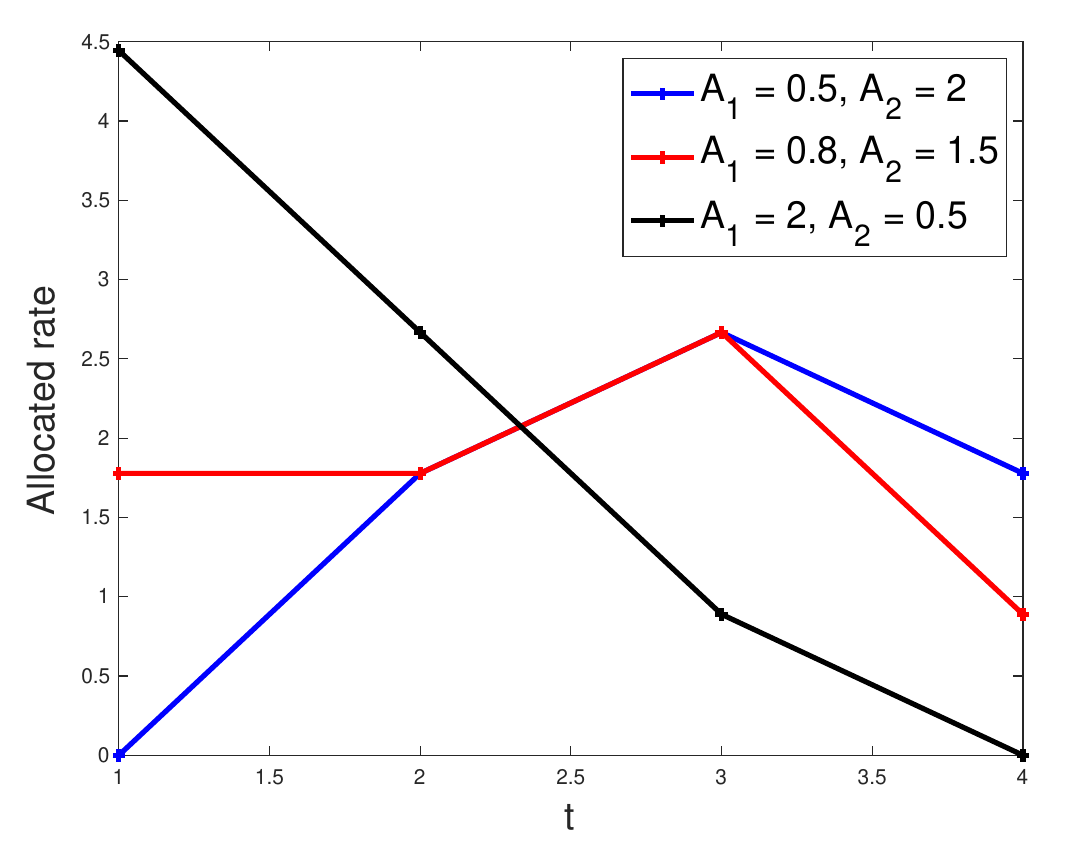}
    \caption{Optimal rate allocation for invariant linear dynamic systems with 4 stages, where $A$ switches from $A_1$ to $A_2$ at time $t=3$. It indicates that when the influence of state information grows (i.e., when the parameter $A$ increases) in a variant dynamic system, more resources should be allocated to the stages near the point of change.}
    \label{fig:rate analysis A variant system}
\end{figure}

\section{Conclusion}
We study goal-oriented compression problems in a linear dynamical system with rate-limited communication resources by considering the long-lasting effect of compression noise on future system states. Our primary objective is to find the optimal transmission rate to reduce the performance degradation induced by compression noise with a given budget. The proposed rate allocation strategy indicates that it is necessary to communicate more in the early stages due to the noise accumulation effect, and also at time slots where there is a noticeable shift of the system dynamics, to mitigate the control action deviation caused by these state changes.


\end{document}